\documentclass[a4paper,final]{amsart}
\usepackage{framed}
\usepackage{graphicx}
\usepackage{multirow}
\usepackage{amsmath,amssymb,amsfonts}
\usepackage{amsthm}
\usepackage{mathrsfs}
\usepackage[title]{appendix}
\usepackage{xcolor}
\usepackage{textcomp}
\usepackage{manyfoot}
\usepackage{booktabs}
\usepackage{algorithm}
\usepackage{algorithmicx}
\usepackage{algpseudocode}
\usepackage{listings}
\usepackage{nicefrac}       
\usepackage{amsmath,amsthm, amsfonts,wrapfig}

\newcommand{\EE}{\mathbb{E}}

\newcommand{\RR}{\mathbb{R}}
\def\dd{\mathrm{d}}
\newtheorem{theorem}{Theorem}

\newtheorem{corollary}{Corollary}

\newtheorem{lemma}{Lemma}
\title{Repeated  Bidding with Dynamic Value}
\author{Benjamin Heymann, Alexandre Gilotte, Rémi Chan-Renous}
\thanks{The authors are with Criteo AI Lab, Paris, France}
\begin{document}

\maketitle

\begin{abstract}

We consider a repeated auction where the buyer's utility
for an item depends on the time  that elapsed since his last purchase.
We present an algorithm to build the optimal bidding policy, and then, because optimal might be impractical, we discuss the cost for the buyer of limiting himself to shading policies. 

\end{abstract}

\section{Introduction}
\subsection{Repeated auctions with dynamic value: a planning problem}
 This article  analyses a \textit{repeated second price auctions} where the buyer's utility for the item depends on the time elapsed since his last purchase of a similar item:  an item value depends on the \emph{age} of the previous purchase. Throughout the paper, We use ad auctions as a motivating application.
The rise of  online marketplaces and digital advertising have fuelled the study of \textit{repeated auctions} along several research axis, in particular budget/ROI constraints, learning, and strategic interactions.
However, surprisingly, we enter into much less explored territory if the buyer's utility for an item \textit{depends on the previous auctions outcomes}, which is notwithstanding a reasonable belief for use cases such as digital marketing, where it can be beneficial to space the marketing interactions over time because of the user's \textit{display fatigue}.
Yet, finding an optimal bidding policy in this widespread setting has not been done yet,  even when the auction is second price. 
This article aims to fill the gap.
We display in Fig.~\ref{fig:value_dynamics_example} an example of  valuation dynamic: after an auction is won, the value of the next auctions drops and then increases back to its initial value as time goes by. 
At this point, it should be intuitive to the reader that bidding the value (the \textit{greedy} strategy) is suboptimal, even though the auction is second price. We illustrate this by showing in Fig.~\ref{fig:partial_gamma} the payoff of several linear scaling strategies. The greedy strategy corresponds to a scaling factor of 1 and is, in our context, \textbf{never} optimal.
 
This is in sharp contrast with the widespread belief that good enough feature engineering  does the job. This intuition  is flawed because it overlooks the underlying \textit{planning problem}. 
Indeed, in the minimal model we introduce,  the bidder has access to the exact immediate value of the item, but we still observe that bidding this value is suboptimal, even in the second price setting.
Last, we mention that in the edge case where the buyer is interested in  buying at most  one item, we get a setting close in spirit to~\cite{livanos2022prophet}, who introduce a  procurement version of the prophet problem, where the goal is to minimize the procurement cost.
\begin{figure}
  \begin{center}
    \includegraphics[width=0.48\textwidth]{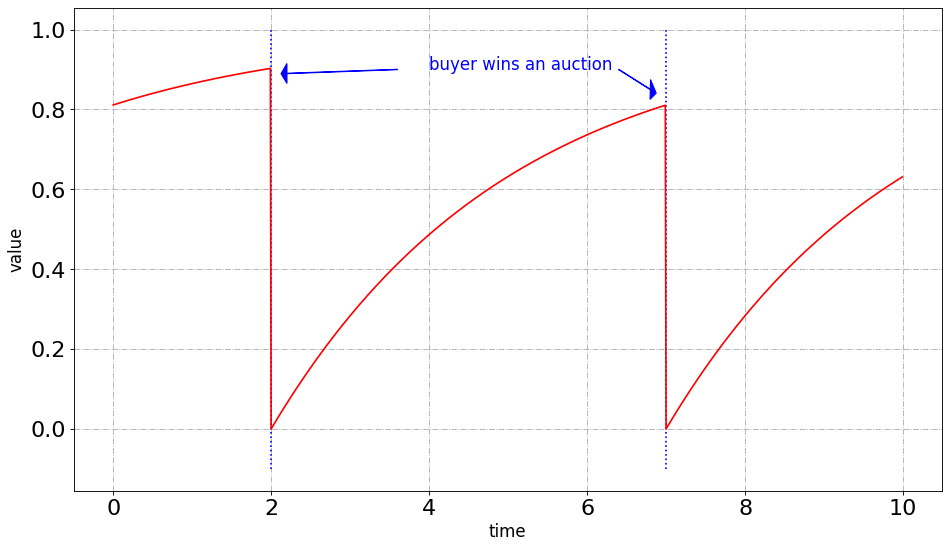}
  \end{center}
  \caption{Example of  dynamical value. When the bidder wins an auction, his value drops brutally, and then progressively goes back to a maximum value. An optimal bidding strategy should account for this phenomenon.}
           \label{fig:value_dynamics_example}
\end{figure}
\subsection{Display advertising auctions}
We proceed with some contextual elements on our  motivating use case,  display advertising auctions.
Digital advertising allows the monetizing of publisher content  programmatically.
When a user reaches one of the publisher's pages, the right to show a display banner to the user is sold. 
The mechanism that elicits the deal is often an auction that takes place in real time, hence the name: Real-Time-Bidding (\textbf{RTB}).
Real-Time-Bidding resents several practical challenges that have motivated a growing body of scientific work at the frontier of mathematics, economics, and computer sciences~\cite{choi2020online}.
The buyer typically derives a valuation formula for the display opportunity of the shape~\cite{bompaire2021causal,balseiro_ec_2017,choi2020online} $ \textit{ Value\_Per\_Event} \times  \Pr(E|X)$,
where $\Pr(E|X)$ is the probability that an event will take place in the future (such as click, sales\ldots it depends) knowing what is available in the user's context (here denoted by $X$) and \textit{Value\_Per\_Event} is a  multiplier that is independent of $X$. 
 This multiplier often integrates business-related constraints such as budget~\cite{balseiro_ec_2017} or ROI or cost per action.  Such a valuation then becomes the input of a bidding module, tasked to find a bid that maximizes the immediate payoff (\textit{i.e.} expected value minus cost). For a second-price auction, this bidding module simply returns the estimated value. For non-incentive compatible auctions such as first-price auctions, methods to estimate the competition have been proposed~\cite{Zhou_2021}. 
  
\paragraph{User fatigue and frequency capping}
One specificity of display advertising is that bidders typically receive a sequence of auctions for opportunities to display an ad to the same user. Those auctions arise as the user browses the web and opens new pages. 
However, it is usually not seen as desirable to display too many similar ads to the same user in a short period of time, and a very common practice consists in defining a \textit{frequency capping }\cite{choi2020online}, which is a maximum number of ads which can be displayed to the same user for the same campaign on a predefined period of time (such as ``maximum 10 displays per day per user"). The bidder would thus stop bidding on a user when the maximum number of ads is reached, until the next period.
 The main rational for applying such frequency capping is essentially that ads displayed to the same user usually have diminishing return. Those return could even become negative, when too many similar ads in a short period of time could result in a degraded user experience. Several studies also point  to the usefulness of spreading the ads in time \cite{sahni2015effect, braun2013online}.
\paragraph{ Fatigue as a feature of the predictors }
To deal with this decreasing return of the ads,  another frequent method consists in  directly modelling how the \emph{past} displayed ads affect the value of the current auctioned ad opportunity. This can be done, for example by measuring the \textit{fatigue} as the number of ads displayed to the user in a past period of time, and using this fatigue feature as a predictor of the $\Pr(E|X)$ model. Several works point to a similar idea, such 
 as \cite{agarwal2009spatio,ma2016user,aharon2019soft}.
We would like to note here that while these features improve the immediate value estimation, they do not account for the impact of an ad on the future opportunities, that is, they still miss the planning problem.
\paragraph{ Impact of winning one auction on future opportunities }
Indeed, one logical consequence of the diminishing return of ads is that winning an auction \emph{now} impacts the value of the \emph{next} opportunities for the same user.
But while it is possible that the $\Pr(E|X)$ term of the valuation recognize the fatigue effect, there is a hole in the literature on how  the bidder component should address the problem. The problem is recognized in~\cite{bompaire2021causal} but no solution is proposed. 
 In a related vein, \cite{diemert2017} propose a heuristic factor to decrease the bids after a click, but the method does not fully address the problem.
In this paper, we propose to study how the bidder should optimally bid in a simplified model inspired by this setting.
We note that the frequency capping could be interpreted as setting the value of the items to 0 during some time after the capping is reached. We  simplify this by assuming that the value of the auctioned item depends solely on the time elapsed since the last won auction.
We also note that today most RTB auctions are first-price. We decided  to study the case of second price auction instead; both for simplicity, and to emphasize that even in this simpler second price setting the problem is non-trivial.
\subsection{Repeated auctions}
 
Several streams of literature related to repeated auctions have flourished around the online marketing use case. Without aiming at being complete, we mention some of them in this section and discuss how the \textit{repetition} of the auction structure the problems under scrutiny.
First, the existence of a budget or an ROI constraint over a sequence of auctions couples all the individual bidding problems together. 
In the literature, such constraints are typically analyzed   either by using a constraint relaxation or by solving an associated knapsack problem. 
The typical solution is then  a scaling  of the value with a well-chosen factor. See for instance 
\cite{balseiro_ec_2017,heymann2019cost,conitzer_or_22,gao2022bidding,castiglioni2022a}.
In this stream of research, this is mostly the coupling of the bidding problems that is under scrutiny. In this line of research, \textit{repetition} is mostly about coupled optimization problems.
Second, the level of bid can be seen as a decision that is repeated over time. 
The outcome (value and payment) depends on this decision. 
This setting can be studied with the bandit framework. See for instance
\cite{weed2016online,achddou2021efficient}.
In this stream of research, this is mostly the explore-exploit tradeoff of the learning bidder that is under scrutiny, and \textit{repetition} is mostly about learning.
Third, there is a line of research that puts under scrutiny the interactions between the seller and the multiple buyers. In this steam of work,  \textit{repetition} is mostly about playing a repeated game
~\cite{amin2014repeated,Drutsa17,golrezaei2019dynamic,Robust_Repeated_First_Price_Auctions,kanoria2020dynamic,nedelec2022learning,deng2022posted}.
\subsection{Contribution}
The applied work~\cite{bompaire2021causal} recognizes the difficulty to handling the coupling between present and future bidding decisions.
We introduce a realistic minimal  model to study this coupling, and frame the design of a bidding policy as a continuous time optimal control problem over an infinite horizon. 
We present structuring properties of this control problem, such as the monotony of the bidding policy and a useful identity for the expected payoff. 
\begin{wrapfigure}{h}{0.4\textwidth}
  \begin{center}
    \includegraphics[width=0.48\textwidth]{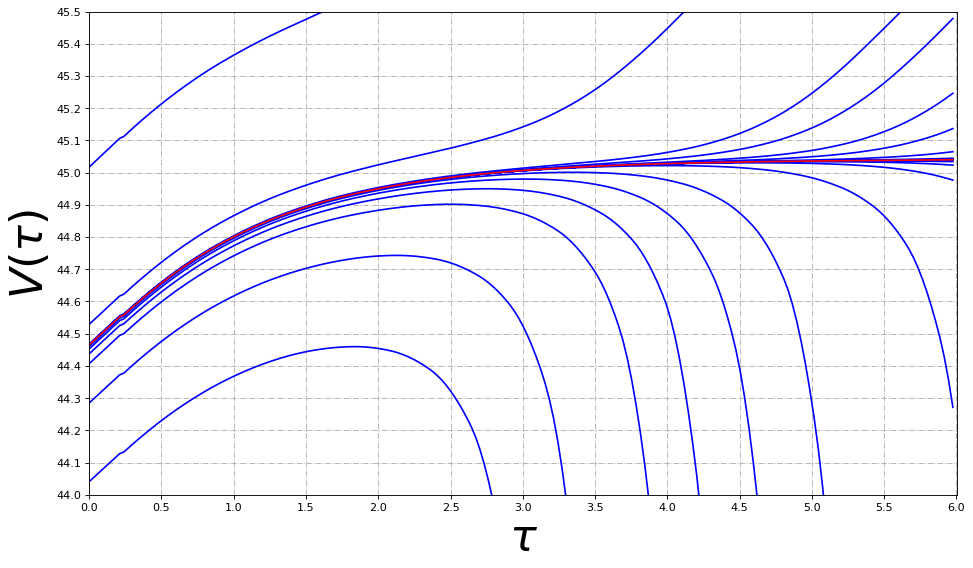}
  \end{center}
  \caption{An example of Algorithm~\ref{algo1} run. In red is the output of the algorithm, in blue the iterates. }
         \label{fig:algorithm}
\end{wrapfigure}
We then introduce an algorithm that iterates toward  the optimal policy. 
Interestingly, the proof relies on a dynamical system argument, more precisely, a refined version  of the Cauchy-Lipschitz  Theorem.
The algorithm allows us to compare through Monte-Carlo simulations the optimal bidding policy and the greedy policy, which corresponds to the economic cost of not accounting for future opportunities in the bidding policy design.
We then take a look at shading policies because they constitute a natural class of simple approximations to look at. 
While we show that shading policies cannot be solutions to the optimal control problem, we also observe that in the numerical experiments, they  perform very well. 
We also derive a closed form expression of the buyer payoff in a specific case. 
In what follows, we  use the masculine pronoun \textit{he} to refer to the buyer, and the feminine pronoun \textit{she} for the seller.
Proofs are postponed to the supplementary material.
\section{Model and bidding problem} 
\label{sec:model-description}
This section introduces a simple environment where the value of an item depends on the time since the last won auction, that we  simply refer to as the system \textit{age} thereafter. 
The section ends with Lemma~\ref{lemma:ODE} that characterizes the value function with a  differential equation.
An auctioneer receives items that she needs to sell immediately to a set of  buyers, and we take the perspective of one of those bidders facing a stationary competition.
The items arrivals follow a Poisson process of intensity $\mu$.  We denote by $\tau$ the  \emph{age}, which is the time that has elapsed since the last won auction.
We assume that the value of an item for the bidder depends solely on the age $\tau$, and  note this value $k(\tau)$. We also suppose that $k$ is a non-decreasing and bounded function.
Fig.~\ref{fig:value_dynamics_example} shows an example of how the value of buying a new item may evolve over a  timeline.
We suppose the auctions to be second price, with an iid competition with a known CDF, $q$.
This competition may include a reserve price set by the seller. In particular, the story includes the posted price scenario.
Thus, $q(b)$ is  the probability that the buyer wins with a bid equal to  $b$. 
Finally, we assume that the sequence  ends at a random time $T_\gamma$, where $T_\gamma$ follows an exponential distribution of parameter $\gamma$, or equivalently we say that the buyer has a discount rate $\gamma$. 
The bidder's goal is to maximize, in expectation, the sum of the values of the auctions he wins, minus the cost he pays, before the end of the sequence. 
Here, we note that this environment may be viewed as a continuous Markov process, where the state is fully summarized by the age $\tau$.
We denote by 
 $p(b)$ the average payment of the user when bidding $b$. Because the auction is second price, we have the relation
$p(b)=q(b)b-\int_0^b q(t)\dd t$.
We denote by $U(v,b) \stackrel{def}{=} q(b)\cdot v - p(b)$  the expected utility when bidding $b$ and valuing $v$ for winning. 
 
Because the state depends solely on the age $\tau$, the bidder's policy can be fully described by a measurable function $ \tau\longrightarrow b(\tau) \in \mathbb{R}^+$ specifying the bid  when the age is $\tau$. We note $\mathcal{B}$ the set of these bidding functions.
Then, for a bidding function $b \in \mathcal{B}$, we  note  $V_{b,\gamma,\mu}(\tau)$ the expectation of the bidder's future revenue when the state is $\tau$, or $V_{b}(\tau)$ for simplicity,
when the context is clear.
Let $T_1 ,T_2 \ldots T_n\ldots $, the time of the next auctions,  $C_1 ... C_n ...$ the competition at these times, and $\tau$ the current age, then
$
V_{b}(\tau) \stackrel{def}{=}  \EE \sum_{i = 1}^{\infty} e^{-\gamma T_i}  \left( k(\tau(T_i)) - C_i\right) 
\textbf{1} \{ b( \tau(T_i) ) > C_i  \}$. 
We are ready to state Lemma~\ref{lemma:ODE}.
\vspace{5mm}
\begin{lemma}
\label{lemma:ODE}
For $b\in\mathcal{B}$,
   $V_b$  satisfies the  differential equation
    \begin{align*}
        V'_b( \tau)   &= - \mu \cdot U \big( k(\tau) + V_b(0) - V_b(\tau) ,  b(\tau) \big)   + \gamma V_b( \tau ).              
    \end{align*}
\end{lemma}
\noindent The proof of Lemma~\ref{lemma:ODE} relies on the dynamic programming principle applied to $V_b$.
\section{Optimal bid}
\label{section:optimal-bid}
This section contains most of the theoretical results. 
We define the optimal value function (or Bellman value), $V^\star(\tau)$, as the sup, on all possible bidding policies, of  $V_b(\tau)$, namely
$V^\star(\tau) \stackrel{def}{=} \sup_{b \in \mathcal{B}} V_b(\tau). $
It is practical to introduce the notation $\pi(v)$  for the expected profit of a bidder of value and bid $v$ in a one shot auction, namely $\pi(v)\stackrel{def}{=}U(v,v)=q(v)v  - p( v)$. This corresponds to the expected profit of the incentive compatible bid in the one shot auction. 
Lemma~\ref{lemma:b_star} provides a  dynamic programming characterization of the Bellman value, and provides a relation between the Bellman value and the bid.
The Bellman value is then proved to be the solution of a parametrized differential equation of unknown parameter (Lemma~\ref{lemma:cauchy}). Theorem~\ref{th:monotony_of_b} ensures that when $k$ is concave, the policy is monotone increasing. 
Theorem~\ref{th:monotony_of_b} can be thought of as a well posedness sufficient condition: primitives for which the optimal bid would not be monotone could arguably be considered pathological from an econometric standpoint.
For   Section~\ref{sub:algo}, that presents the algorithm for computing an optimal policy, we  suppose that $q$ is continuous, and thus $\pi$  is $C^1$. \footnote{Indeed, by definition, 
$
    \pi(v)=q(v)v  - p( v)=\int_0^v q(t)\dd t
$}
\subsection{Differential Equation}
From the proof of Lemma~\ref{lemma:ODE}, we get the following characterizations of $V^\star$ and the optimal bid $b^\star$
\vspace{5mm}
\begin{lemma}
\label{lemma:b_star}
We have the relation
$V_{t}^\star =\int_0^{+\infty}\mu e^{-(\mu+\gamma)t}
   \Big( \pi(k_t+V_0^\star-V_t^\star) +V_t^\star\Big)
    \dd t
$. Moreover, there exists an  optimal bid $b^\star$, and it satisfies the relation 
\begin{align*}
   b^\star(\tau) = \max \left(0;  k(\tau) + V^\star(0) - V^\star(\tau)  \right) \ . 
\end{align*}
\end{lemma}
\vspace{5mm}
This differential equation is complicated to solve directly because its initial value $V^\star_0$ also appears as a parameter of the dynamic. To clarify this point, we define
 $\Phi: \RR \times \RR_+ \times \RR_+\to \RR$ such that 
$
  \Phi(t, v, \lambda) =\gamma v -\mu\pi(k_t+\lambda-v).    
$
We then reformulate lemma \ref{lemma:b_star} as a more usual Cauchy problem:
\vspace{5mm}
\begin{lemma}
\label{lemma:cauchy}
   The value function $V^\star$ is the solution of the ordinary differential equation 
\begin{align*}
\label{eq:system}
\tag{$\mathcal{F}_{y_0}$}
    \left\{
    \begin{array}{ll}
        \dot{Y}_t = \Phi(t, Y_t, y_0) \\
        Y_0=y_0 
    \end{array}
\right.
\end{align*}
   for some $y_0\in \RR_+$.
\end{lemma}
\vspace{5mm}
\noindent It should be noted that in Lemma~\ref{lemma:cauchy}, the parameter $y_0$ is not given.
\subsection{Differential equation for $b$}
In this section, we provide a parametrization of the differential equation so that the $b^\star$ appears as solution.
Such parametrization requires additional assumptions on $k$.
The optimal bid also follows a differential equation:
\vspace{5mm}
\begin{lemma}
\label{lemma:b_dynamics}
If $k$ is $C^1$, then, on every  open on which  $b^\star(t) > 0$, the optimal bid $b^\star$ satisfies the differential equation
    $\dot{b}_t = \dot{k}_t-\gamma(V^\star_0+k_t) +\mu\pi(b_t)+\gamma b_t$. 
\end{lemma}
\vspace{5mm}
Since $k$ is increasing with $\tau$, it may seem natural to expect that the optimal bid would also increase with $\tau$.
The following result confirms this intuition under the additional assumption that  $k$ is concave.
\vspace{5mm}
\begin{theorem}
\label{th:monotony_of_b}
If $k$ is concave, then $b^\star$ is increasing with $\tau$, and strictly increasing on any interval where $k$ strictly increases. 
\end{theorem}
\vspace{5mm}
\noindent  We insist that this is not true without the concavity hypothesis: Fig.~\ref{fig:counter-example} proposes a counter example where $k$ is increasing but not concave.
\subsection{Algorithm}
\label{sub:algo}
Lemma~\ref{lemma:cauchy} does not allow us to derive $V^\star$ from a Cauchy problem because the  initial condition is also a parameter of the dynamics. 
Still,  by the Cauchy-Lipschitz Theorem, the solution of the ordinary differential equation 
\begin{align*}
\tag{$\mathcal{F}_{y_0,\lambda}$}
    \left\{
    \begin{array}{ll}
        \dot{Y}_t = \Phi(t, Y_t,\lambda) \\
        Y_0=v_0
    \end{array}
\right.
\end{align*}
admits a unique maximal solution $Z^{y_0,\lambda}:t\to Z^{y_0,\lambda}(t)$ for any  $y_0>0$ and $\lambda>0$.
We also observe that the problem $\mathcal{F}_{y_0,y_0}$ is the same at the problem of equation $\mathcal{F}_{y_0}$, and we define $Z^v(t) \stackrel{def}{=} Z^{v,v}(t) $ for all $t\geq 0$.
In the following, we  assume that $q$ is continuous, and thus $\pi$  is $C^1$. This allows to apply the  differentiable dependency Theorem, which is the key to our approach.
Lemma~\ref{lemma:algo} is the most technical result of the paper,  
and is leveraged in the proof of Theorem~\ref{theorem:algo}.
\vspace{5mm}
\begin{lemma}
\label{lemma:algo}
Suppose $q$ is continuous.
The value
$V_0^\star$ is the unique  $v$ for which $\lim_{t\to+\infty} Z^{v}(t) )$ is finite. 
\end{lemma}
\vspace{5mm}
We then use Lemma~\ref{lemma:algo} to prove that a simple dichotomy on the parameter allows us to solve the bidding problem for repeated auctions with dynamic value. 
\vspace{5mm}
\begin{theorem}[Algorithm]
\label{theorem:algo}
Suppose $q$ is continuous.
Let $y^{(n)}$  the iterations of  Algorithm~\ref{algo1}, then 
$
    ||V_0^\star-y^{(n)}(0)||= O(1/2^n).
$
\end{theorem}
\vspace{5mm}
\begin{algorithm}
\caption{Computing $V_0$}\label{algo1}
\begin{algorithmic}[1]
\Require $(a,b,N)$
\For{$n \leq N$}
        \State $c\leftarrow (a+b)/2$
        \State Compute $Z^{c,c}$, solution of $\mathcal{F}_{c}$
        \State $y^{(n)} \leftarrow Z^{c,c} $
\If{$t\to y^{(n)}(t)$}
\State $a\leftarrow c$
\Else
\State  $b\leftarrow c$ 
\EndIf
\EndFor
\end{algorithmic}
\end{algorithm}
\subsection{Numerical estimation of the cost of impatience}
Algorithm~\ref{algo1} is implemented in Python 3. 
We use \textsc{solve\_ivp}, an ODE  solver from \textsc{Scipy}~\cite{2020SciPy-NMeth} to numerically solve the ODE. We rely on the default Runge-Kutta method~\cite{dormand1980family,Lawrence1986SomePR}.
To test the different bidding policies, we use Lemma~\ref{lemma:solution_ode_bid} in the Appendix to say that, for $\gamma\to 0 $, we can compute the time average performance over a long enough period. 
For the experiments, reported in Table~\ref{table:cor}, we used $\gamma=0.01$ to construct the optimal policy. 
We ran simulation on this optimal policy and on the greedy policy. 
We tested several values for $\mu$, with $k(\tau)=1-e^{-\tau}$ and $k(\tau)=1-1/(1+\tau)$.
\begin{table}
\parbox{.45\linewidth}{
\begin{tabular}{@{}llllll@{}}
\toprule
$\mu$   & 0.1   & 1    & 5    & 10   & 100  \\ \midrule
optimal & 0.045 & 0.22 & 0.43 & 0.52 & 0.72 \\
greedy  & 0.043 & 0.20 & 0.34 & 0.38 & 0.46 \\ \bottomrule
\end{tabular}
}
\quad 
\parbox{.45\linewidth}{
\begin{tabular}{@{}llllll@{}}
\toprule
$\mu$   & 0.1   & 1    & 5    & 10   & 50   \\ \midrule
optimal & 0.035 & 0.16 & 0.32 & 0.40 & 0.69 \\
greedy  & 0.035 & 0.14 & 0.26 & 0.31 & 0.44 \\ \bottomrule
\end{tabular}}
\caption{ Estimated value per time unit with the greedy bidder and with the optimal bidder, for different values of $\mu$, and $\gamma$ close to 0 with $k(t)=1-e^{-t}$(left) and $k(t)=1-1/(1+t)$ (right) }
\label{table:cor}
\end{table}
\section{Shading strategy}
\label{section:shading-strategy}
A shading strategy consists in bidding a constant factor $\alpha$ (smaller than 1) of the value, \textit{i.e.} $b(\tau)= \alpha k(\tau)$. This class of strategy is of great practical and theoretical interest because it is simple to implement and to analyze. It also satisfies some theoretical properties. 
For instance, as mentioned in the introduction,  shading is used as a way to implement budget or ROI constraints. 
We  use in this section the shorthand notation $V_\alpha = V_{\alpha k,\gamma,\mu}$.
We found it notable that on the settings we tested, the class of shading strategies performs very well. It can be proven however that it is not optimal in general. 
Indeed, by Lemma~\ref{lemma:b_dynamics}, the optimal bidding policy satisfies
$\dot{b}_t = \dot{k}_t-\gamma(V^\star_0+k_t) +\mu\pi(b_t)+\gamma b_t$, which implies that  a shading policy $b_t = \alpha k_t$  that would outperform all other strategies  should satisfy
$\alpha \dot{k}_t = \dot{k}_t-\gamma(V^\star_0+k_t) +\mu\pi(\alpha k_t)+\gamma \alpha k_t 
$, which in general does not hold. 
\subsection{Special case}
We mention that there is an instance of the problem where the time average payoff of a shading strategy can be computed explicitly.
We use Theorem~\ref{th:special_case} to build a visualization of the average expected payoff, that we display in Fig.~\ref{fig:partial_gamma}. We check that shading strategies outperform truthfulness in
the presence of a value dynamic. We also check that the greater the arrival rate $\mu$, the stronger  the shading  should be (the $\alpha$ multiplier close to zero).
\begin{theorem}
\label{th:special_case}
If the competition is a uniform distribution on $[0,1]$, if $k_t=1-1/(1+t)$ 
then for all $\alpha$ satisfying $\mu\alpha>1$
\begin{align*}
 \lim_{\gamma\to 0 }  \gamma V_0^\alpha   = ( 1 - 1/2\alpha )(\mu\alpha)^{2}\frac{
\Gamma(\mu\alpha-1,\mu\alpha)}{\Gamma(\mu\alpha+1,\mu\alpha)},
\end{align*} 
where $\Gamma(\_,\_)$ is the upper incomplete gamma function.\footnote{defined as $\Gamma(s, x)=\int_x^{\infty} t^{s-1} e^{-t} \dd t$}
\end{theorem}
\begin{wrapfigure}{h}{0.4\textwidth}
  \begin{center}
    \includegraphics[width=0.48\textwidth]{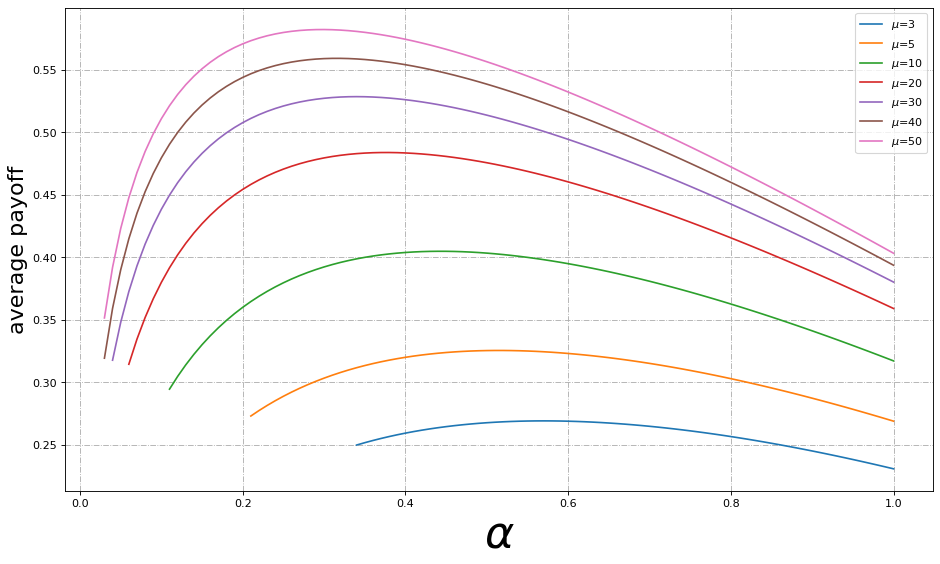}
  \end{center}
  \caption{Illustration of Theorem~\ref{th:special_case} }
         \label{fig:partial_gamma}
\end{wrapfigure}
\subsection{Numerical comparison}
Figure \ref{fig:perf_shading_exp} displays the proportion $\nicefrac{V_\alpha}{ V^\star }$  of the optimal value retrieved by bidding $\alpha k$, for several values of $\alpha$ and $\mu$, and two distinct concave functions $k$ (and with $\gamma$ was close to 0). We note that on all these settings, the best value of $\alpha$ gives results very close to the optimal value. (it was higher than $99\%$)
Another observation is that for smaller values of $\mu$,  
the greedy bidder (\textit{i.e.} $\alpha = 1$) performs near optimally. This observation is easily explained: a small $\mu$ means that there are few auctions per time unit; it is thus likely that the first auction only arrives a long time after $0$ , when the function $k(t)$ is already close to its sup. This is thus not so different from the case when $k$ is constant, for which the greedy bidder is optimal.
On the other hand, for larger values of $\mu$, the greedy bidder performs poorly; and the best shading factor decreases with the density of auctions $\mu$.
\begin{figure}
\begin{center}
\includegraphics[width=0.45\textwidth]{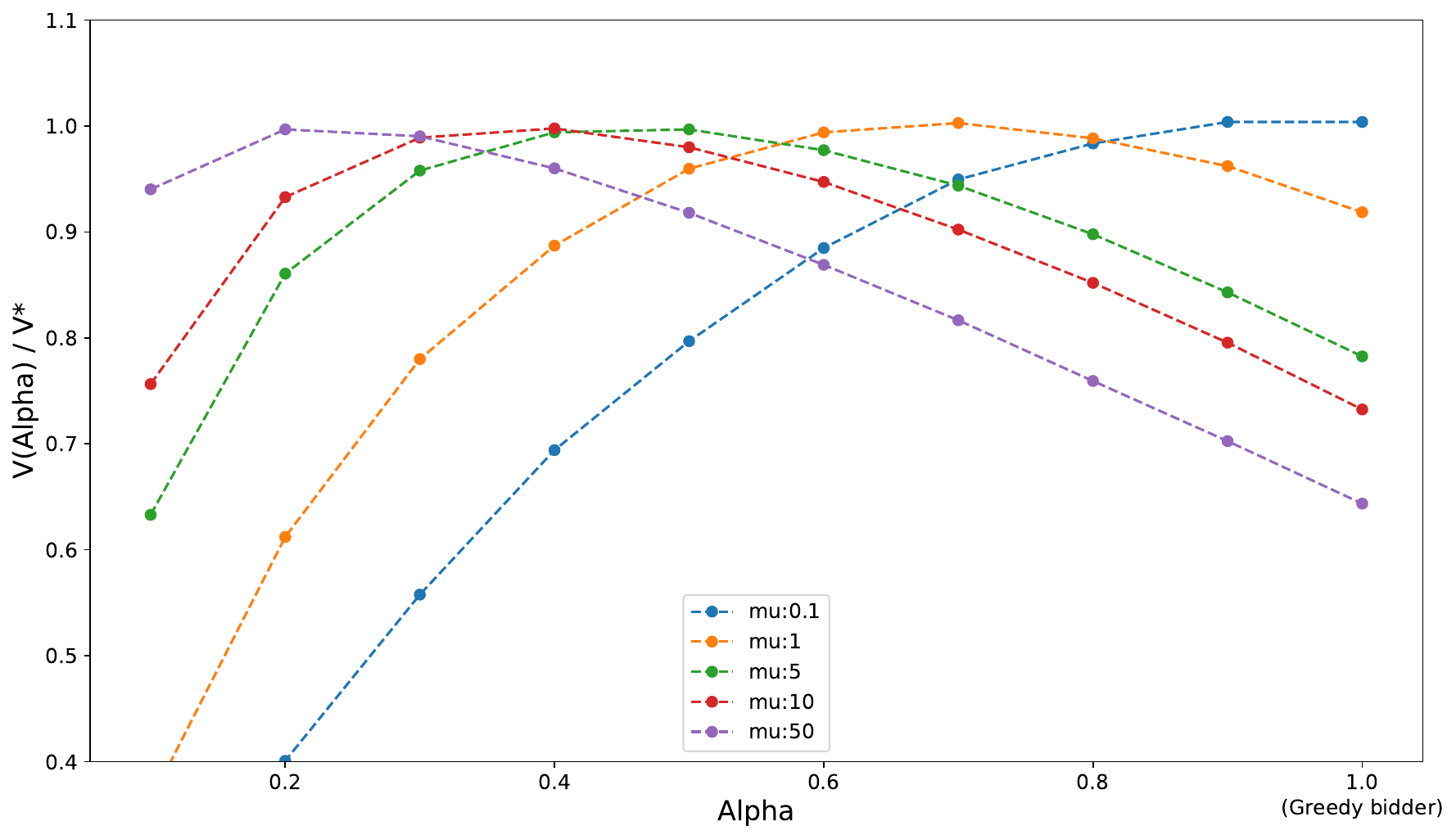}
\includegraphics[width=0.45\textwidth]{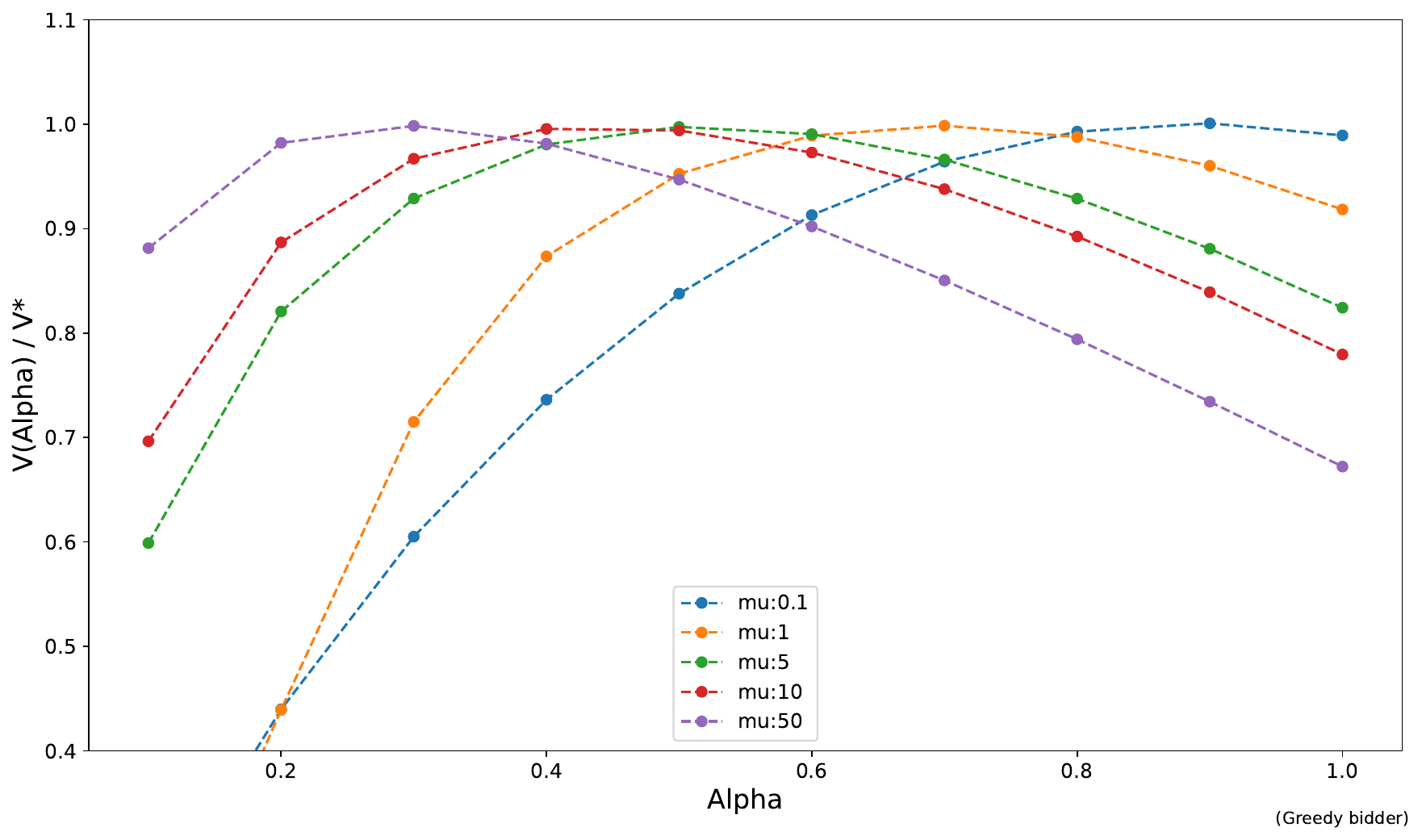}
\end{center}
\caption{Ratio $\nicefrac{ V_\alpha }{ V^\star }$ as a function of $\alpha$ with $k_\tau = 1-e^{-t}$ and $\mu = 5$ (left) and with $k(t) =  1 - \nicefrac{1}{1+t} $ }
\label{fig:perf_shading_exp}
\end{figure}
\paragraph{Example with $k(\tau)$ non concave}
We have seen in Lemma~\ref{lemma:b_star} that $b^\star$ is an increasing function when $k$ is concave. We propose here an example where $k$ is not concave to observe what can happen in this case. The function $k$ we used is depicted in Fig.~\ref{fig:counter-example}.
It is a simple piece-wise linear function, making two steps. The optimal bids as a function of $\tau$ are plotted for different values of $\mu$: we can observe that they are not monotonous. The intuition for this behavior is simple: since the function $k$ increases quickly near 0, it is worth winning early because this generates a high value per time unit. Thereafter, the function stops increasing until the second step, where it increases very sharply again. Just before the increase, it is better not to bid at all : the bidder should indeed avoid resetting the state to 0 just before this strong increase.
Fig.~\ref{fig:counter-example} shows the value obtained when submitting a shaded bid $k(t)\times\alpha$ , with the same two steps k function. We note that the best $V_\alpha$ here is quite sensibly smaller than $V^\star$.  We also note that for large values of $\mu$, $V_\alpha$ is not a concave function of $\alpha$ and may have several local maxima. 
This example also contradicts another assumption which could seem intuitive: $b^\star$ is \emph{not} an increasing function of $k$.  Indeed, if we removed the second step on the $k$ function above, then we would get a concave function, and the optimal bid would be positive for all $t>0$.
\begin{figure}
\begin{center}
\includegraphics[width=0.45\textwidth]{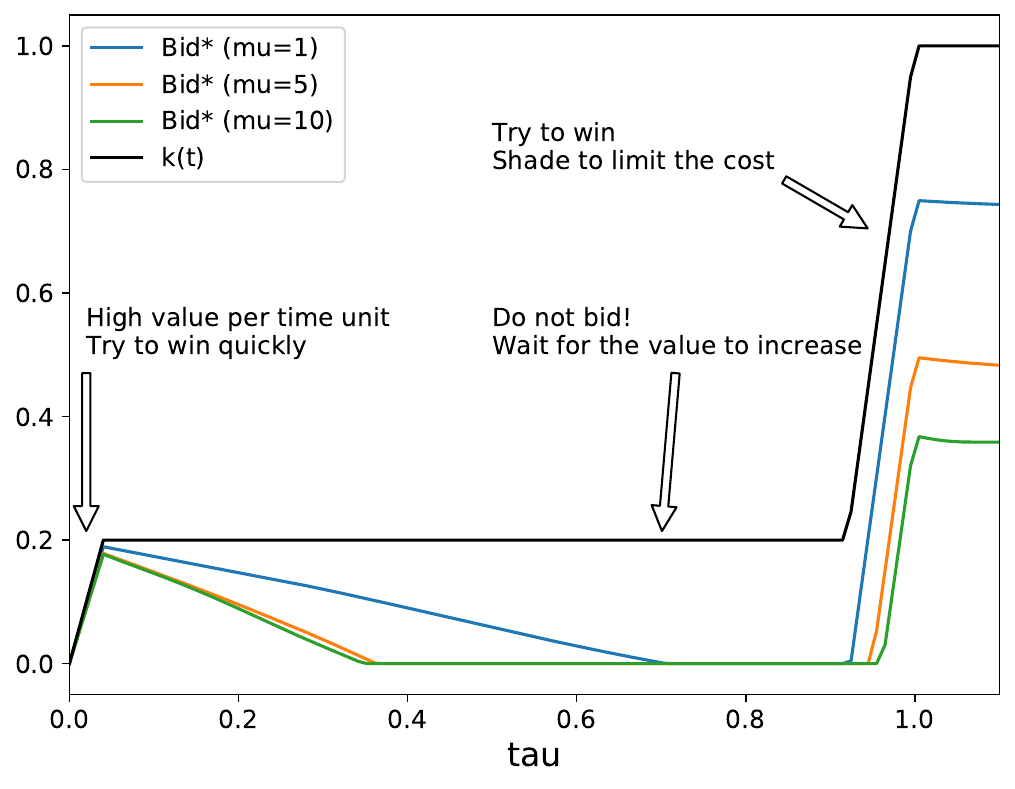}
\includegraphics[width=0.45\textwidth]{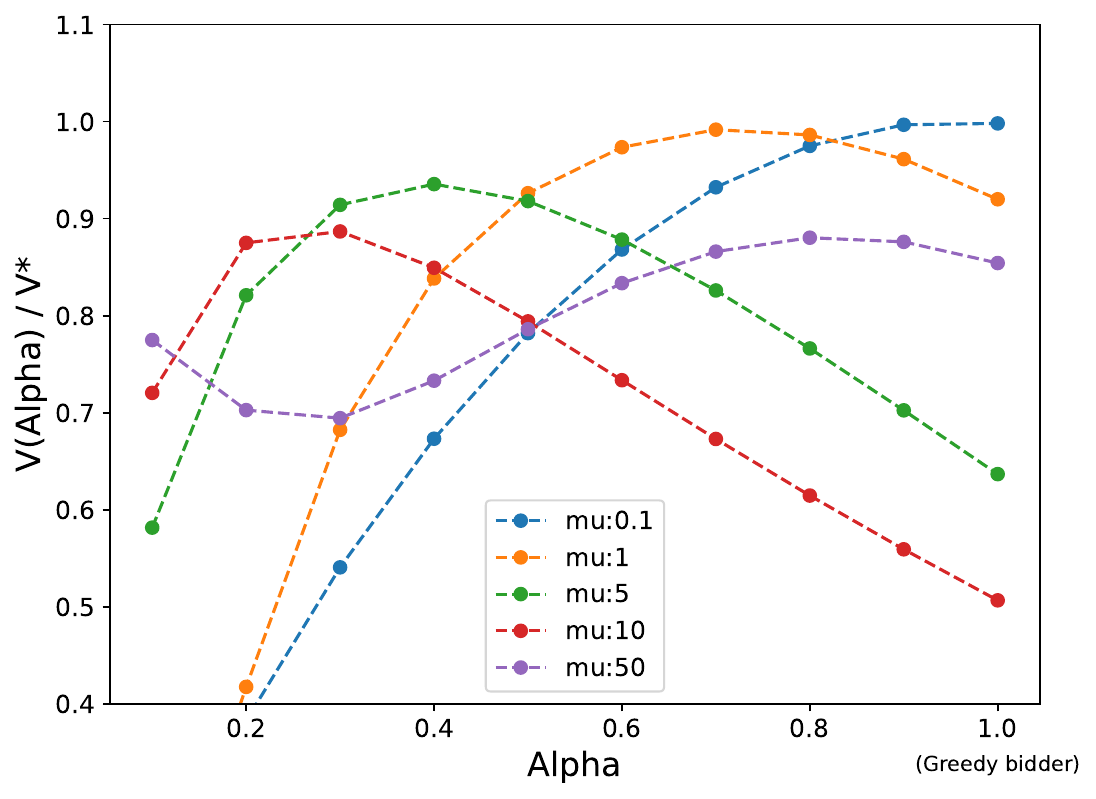}
\end{center}
\caption{(left)  $k$ and optimal bid as a function of $\tau$, for different values of $\mu$ (right). Optimal bids and $V_\alpha$ with the two steps function $k(t) = \min( 5t , 0.2 ) + \max(0,\min( 0.8, 10t - 9.2 ) $}
\label{fig:counter-example}
\end{figure}
\subsection{Asymptotic convergence}
Theorem~\ref{theorem:limit} provides  a complementary perspective on why  shading policies might work well.
Basically, as $\mu$ goes to $+\infty$, the best shading policy performs closer and closer to the optimal policy. 
\begin{theorem}
\label{theorem:limit}
Set $\mathfrak{p}(b) =\EE(C|C\leq b)$, where $C$ is the price to beat.
Suppose $k$ concave and such that $k\in C^1$  and  $k_0=0$.
Suppose $q\in C^1$, and $\dot{q}(0)>0$,
then 
\begin{align*}
  \sup_{b\in\mathcal{B}} \lim_{\mu\to +\infty}\lim_{\gamma\to 0 } \frac{V_{b,\gamma,\mu}(0)-  V_{k,\gamma,\mu}(0) }{ V_{b,\gamma,\mu}(0)}=
  \dot{\mathfrak{p}}(0)
   = \lim_{t\to 0}\frac{\dot{q}(t)\int_0^tq(s)\dd s }{q(t)^2}\ ,
\end{align*}
Moreover,
\begin{align*}
\sup_{b \in \mathcal{B}} \lim _{\mu \rightarrow+\infty} \lim _{\gamma \rightarrow 0} V_{b, \gamma, \mu}(0)
=
\sup _{\alpha \in [0,1]} \lim _{\mu \rightarrow+\infty} \lim _{\gamma \rightarrow 0}\ .
V_{\alpha \cdot k, \gamma, \mu}(0)
\end{align*}
\end{theorem}
\section{Conclusion}
Auction systems with dynamic values are everywhere.
This work is a first step toward better understanding them. 
Further research could include studying what happens when the bidder does not know the dynamics in advance, and needs to learn, as well as extending the results to more general auctions and more general value dynamics.

\clearpage
\newpage
{\begin{center}\Large\textbf{Appendix \\ -- \\ }\end{center}}\vspace{1cm}
\appendix
\section*{Proof of Section 2}
\paragraph{Expected value when an auction starts in state $\tau$}
Let  $W_{b}(\tau)$ be the expected revenue of the bidder when an auction is received in state $\tau$.
Then $W_{b}(\tau)$ can be written by taking the expectation on the outcome  of the immediate auction.  On average, the bidder pays the second price $p( b(\tau) )$. We then consider two possible outcomes. 
     Either the auction is lost: then the bidder pays 0, stays in state $\tau$, and its expected future reward is $V(\tau)$. 
     Or he wins the auction: the bidder then receives a reward $k(\tau)$ and the state is reset to $0$, which means that the bidder's expected reward after the auction is $V_b(0)$.
Wrapping all this together, we get
\begin{align}
 W_{b}(\tau)  &=  q( b(\tau) ) \cdot \big( k(\tau)  + V_b(0)  ) \big) - p( b(\tau) )  +  (1 - q( b(\tau) )) \cdot V_b(\tau) \nonumber \\
             &=  V_b(\tau) +  q( b(\tau) )  \cdot \big(  k(\tau)  + V_b(0) - V_b(\tau) \big) - p( b(\tau) ) \nonumber    \\
             &= V_b(\tau) + U\big(  k(\tau) + V_b(0) - V_b(\tau) , b(\tau) \big).  \label{eq:w_b} 
\end{align}
\paragraph{Bellman equation}
We note $T$ the waiting time until the next auction. For our stationary Poisson process, $T$ always follows an exponential distribution, of parameter $\mu$. 
We can write the evolution of the system as an expectation on $T$:
\begin{align}
V_b(\tau) =  \EE_T \big(  e^{- \gamma T} \cdot W_b( \tau + T )  \big ) 
          = \int_0^{+\infty}  e^{-(\gamma + \mu) t} \cdot \mu \cdot  W_b( \tau + t ) dt. \label{eq:v_b}
\end{align}
\subsection*{Lemma~\ref{lemma:ODE}}
\begin{proof}
    This is obtained directly by writing the evolution of $V$ from $\tau$ to $\tau + \delta$. 
        Similarly to \eqref{eq:v_b}, let $T$ the time of the next auction, and $\delta$ a positive time. We split the expectation of $V$ in two parts, when $T$ is either less than $\delta$ or larger than $\delta$.
        We have
        $  
        V_b(\tau) = \int_0^{\delta}  e^{-(\gamma + \mu)t} \cdot \mu \cdot W_b( \tau + t ) dt + e^{-(\gamma+\mu)\delta} \cdot V_b( \tau + \delta ).
        $
        From this we deduce that for all $\tau, \delta >0$
        $   
         \frac{ V_b(\tau +\delta) -  V_b(\tau)}{ \delta} = -\frac{1}{\delta} \int_0^{\delta}   e^{-(\gamma + \mu).t} \cdot \mu \cdot W_b( \tau + t ) dt +  \frac{1-e^{-(\gamma+\mu)\delta }}{\delta} V_b( \tau + \delta ).
        $
        The right-hand side is bounded, from which we first conclude that $V_b$ is continuous (and thus $W$ is continuous too); we may now take the limit when $\delta \rightarrow 0$ of the right-hand size, which is $ -\mu \cdot W_b( \tau ) + (\gamma + \mu) \cdot V_b(\tau) $. This proves both the derivability and the formula. 
\end{proof}
\section*{Proof of Section 3}
\subsection*{Lemma~\ref{lemma:b_star}}
\begin{proof}
     If a bidding function achieves the sup in the definition of $V^\star$, then, since the auction is second price, we can increase the expected return at any point $\tau$ by bidding the true value $ k(\tau) + V^\star(0) - V^\star(\tau)$. This implies that if there is an optimal bidding function, it can only be $b^\star$ as defined above.
    The existence of an optimal bidding function is a consequence of the classical result that the Bellman operator is a contraction: For any value function $\tau \to V(\tau)$, we define the bidding policy $B(V) := \tau \to \max \left(0;  k(\tau) + V(0) - V(\tau)  \right) $.
     Because the auction is second price, bidding with $B(V)$ is the optimal action in any state $\tau$ if  the value we get after this auction if defined by $V$.
     From this, we deduce that for any $b$, $ V_{B(V_b)} \ge V_b $, by recurrence on the number of steps where we use $B(V_b)$ instead of $b$.
     We also note that the operator $V \to V_{B(V)}$ is a contraction for the norm $\infty$ (eg: \cite{bertsekas2019reinforcement}), from which we deduce that any sequence defined by $V_{i+1} = B(V_i)$ will converge to the same value. Since this sequence is also increasing, it must converge to $V^\star$, and the associated bids converges to $b^\star := B( V^\star) $.
     We plug $b^\star$ in Lemma~\eqref{lemma:ODE} to get 
$    V_{t}^\star =\int_0^{+\infty}\mu e^{-(\mu+\gamma)t}
   \Big(
   U(b^\star_t,k_t+V_0^\star-V_t^\star)
   +V_t^\star\Big)
    \dd t$.
    This is maximized for $b^\star_t=k_t+V_0^\star-V_t^\star$, in which case, 
    $U(b^\star_t,k_t+V_0^\star-V_t^\star) = \pi(k_t+V_0^\star-V_t^\star)$
\end{proof}
\subsection*{Lemma~\ref{lemma:b_dynamics}}
\begin{proof}
From Lemma~\ref{lemma:b_star}, $b^\star_t = k_t+V_0^\star - V_t^\star$, hence since by Lemma~\ref{lemma:ODE}, $V^\star$ is differentiable, $\dot{b}^\star_t = \dot{k}_t- \dot{V}_t^\star$. 
Replacing in the latter expression $\dot{V}^\star_t$ by $\gamma V_t^\star-\mu \pi(k_t+V_0^\star -V_t^\star)$ from Lemma~\ref{lemma:cauchy}, we get
$\dot{b}^\star_t= \dot{k}_t-\gamma V_t^\star +\mu \pi(k_t+V_0^\star -V_t^\star)= \dot{k}_t-\gamma(k_t+V_0^\star-b_t^\star) +\mu \pi(b^\star_t)$.
Hence $ \dot{b}_t = \dot{k}_t-\gamma(V^\star_0+k_t) +\mu\pi(b_t)+\gamma b_t$.    
\end{proof}
\subsection*{Theorem~\ref{th:monotony_of_b}}
\begin{proof}
    When $k$ is smooth, the proof follows from the fact that $\dot{b}^\star(t)$ is the sum of $\dot{k}_t-\gamma(V^\star_0+k_t)$, which is (strictly) decreasing with $t$ under the hypotheses, and $\mu\pi(b^\star)+\gamma b^\star$ which varies in the same direction as $b^\star$.
We treat the case when $k$ is strictly increasing, (the case $k$ increasing is similar), and we proceed ad absurdum. 
Suppose that $b^\star$ is non-increasing on an interval $]t_1,t_2[$.
Then $b'$ must be strictly decreasing on this interval, and thus $b'(t_2) <0$. This ensures that the interval can be extended, and thus $b^\star$ would be decreasing up to $+\infty$, at a rate at least $b'(t_2)< 0$. This implies that the solution of the differential equation \ref{lemma:b_star} would go to $-\infty$.
We remind here that this differential equation \ref{lemma:b_star} is no longer valid when $b^\star$ takes negative values (the bid cannot become negative). But this still means that the bid would become 0 at some point $t_0$, and the same argument proves it cannot become positive again after that: the bid would thus be 0 on $]t_0;+\infty[$.
But bidding 0 on $]t_0, +\infty[$ would generate a return of 0; which clearly is not optimal.
We conclude that $b^\star$ must be strictly increasing.
For non-smooth $k$, the concavity still implies that $k$ must be $C^1$ with $k'$ decreasing on a dense subset of $R^+$ . The result  still holds by slightly adapting the previous proof, noting that $b'$ is defined on the same set as $k'$, and is equal to $k'$ plus a continuous function. This implies that if $b'$ decreasing and negative just before $t_2$, it must still be negative just after $t_2$.  
\end{proof}
Lemma~\ref{lemma:cauchy} does not allow us to derive $V^\star$ from a Cauchy problem because the  initial condition is also a parameter of the dynamics. 
Still,  by the Cauchy-Lipschitz Theorem, the solution of the ordinary differential equation 
\begin{align}
\label{eq:system_krikorian}
\tag{$\mathcal{F}_{y_0,\lambda}$}
    \left\{
    \begin{array}{ll}
        \dot{Y}_t = \Phi(t, Y_t,\lambda) \\
        Y_0=v_0
    \end{array}
\right.
\end{align}
admits a unique maximal solution $Z^{y_0,\lambda}:t\to Z^{y_0,\lambda}(t)$ for any  $y_0>0$ and $\lambda>0$.
We also observe that the problem $\mathcal{F}_{y_0,y_0}$ is the same at the problem of equation $\mathcal{F}_{y_0}$, and we define $Z^v(t) := Z^{v,v}(t) $ for all $t\geq 0$.
\begin{lemma}[Taylor extension]
\label{eq:lemma_Taylor}
Set $u_t = u_t(v)=k_t+v - Z^v(t)$, then 
    \begin{align*}
       \frac{\partial Z^v(t)}{ \partial v }  =
    e^{\gamma t +\mu\int_{0}^t  q( u_s(v)) \dd s }\big(1-
   \mu \int_{0}^t e^{-\gamma s -\mu\int_{0}^s q(u_l(v)) \dd l }  q( u_s(v)) \dd s\big)
    \end{align*} 
\end{lemma}
\begin{proof}
By the differentiable dependency Theorem~\cite{krikorian}, for any $t>0$, $(y_0,\lambda)\to Z^{y_0,\lambda}(t)$  is differentiable, and its differential maps a perturbation
$(\Delta v, \Delta \lambda)$ to $\Delta Z^{\Delta v, \Delta \lambda}_{y_0,\lambda}$, which is   the solution of 
\begin{align}
\label{eq:system_krikorian_2}
    \left\{
    \begin{array}{ll}
        \dot{\Delta Z}^{\Delta v, \Delta \lambda}_{y_0,\lambda}(t) =D_y \Phi(t, Z^{y_0,\lambda}(t),\lambda).\Delta Z_{y_0,\lambda}^{\Delta v, \Delta \lambda}(t) 
        +D_\lambda \Phi(t, Z^{y_0,\lambda}(t),\lambda).\Delta \lambda  \\
        \Delta Z^{\Delta v, \Delta \lambda}_{y_0,\lambda}(0)=\Delta v. 
    \end{array}
\right.
\end{align}
By definition $Z^v(t) = Z^{.,.}(t)\circ m(v)$ where $m(v) =(v,v)$, so that by derivation of a composition, 
\begin{align*}
    D_v (Z^v(t))[\Delta v] =D_{m(v)} Z^{.,.}(t).  D_v m [\Delta v]  
    =D_{m(v)} Z^{.,.}(t). (\Delta v,\Delta v) 
    =\Delta Z^{\Delta v, \Delta v}_{v,v}(t)
\end{align*}
therefore, using~\eqref{eq:system_krikorian_2},
\begin{align}
\label{eq:1234}
   \dot{D_v Z}^v(t)[\Delta v]=\dot{\Delta Z}^{\Delta v, \Delta v}_{v,v}(t) =
   D_y \Phi(t, Z^{v,v}(t),v).\Delta Z_{v,v}^{\Delta v, \Delta v}(t) 
        +D_\lambda \Phi(t, Z^{v,v}(t),v).\Delta v.
\end{align}
We now compute the two partial derivatives. First
since   $\Phi(t, y, \lambda) =\gamma y -\mu\pi(k_t+\lambda-y)$,
$D_y \Phi(t,y,\lambda)=\gamma + \mu\dot{\pi}(k_t+\lambda-y)$.
Now
since  $\pi(v)  =  q(v)v  -p( v)=\int_0^v q(t) \dd t$ we have $\dot{\pi}(v)=q(v)$, therefore
$D_y \Phi(t,y,\lambda)=\gamma + \mu q( k_t+\lambda-y)$
Similarly, 
$D_\lambda \Phi(t,y,\lambda)= -\mu\dot{\pi}(k_t+\lambda-y)= -\mu q(k_t+\lambda-y)$.
Therefore, combining with~\eqref{eq:1234}
\begin{align*}
  \dot{D_v Z}^v(t)[\Delta v]=
  D_y \Phi(t, Z^{v,v}(t),v).\Delta Z_{v,v}^{\Delta v, \Delta v}(t) 
        +D_\lambda \Phi(t, Z^{v,v}(t),v).\Delta v\\
        =    (\gamma + \mu q( k_t+v-Z^{v,v}(t))).\Delta Z_{v,v}^{\Delta v, \Delta v}(t) 
  -\mu q( k_t+v-Z^{v,v}(t))\Delta v\\
=    (\gamma + \mu q( k_t+v- Z^v(t) )).D_v ( Z^v(t) )[\Delta v] 
  -\mu q( k_t+v- Z^v(t) )\Delta v
\end{align*}
We observe that the equation on $D_v Z^v(t)[\Delta v]$ is of the shape
    $\dot{X}_t = a_t X_t + b_t$ with     $X_0 = \Delta v$.
The solution of such ODE is 
$
    X_t = e^{\int_{0}^t a_s \dd s }(\Delta v + \int_{0}^t e^{-\int_{0}^s a_u \dd u }b_s \dd s).
$
Therefore,
\begin{align*}
    &D_v Z^v(t)[\Delta v]= \\
    &
    e^{\int_{0}^t  (\gamma + \mu q( k_s+v-w_s(v))) \dd s }\Big(\Delta v 
    -
    \int_{0}^t e^{-\big(\int_{0}^s (\gamma + \mu q( k_l+v-w_l(v))) \dd l\big) } \mu q( k_s+v-w_s(v))\Delta v \dd s\Big)\\
    &= \Big(
    e^{\gamma t +\mu\int_{0}^t  q( u_s(v)) \dd s }\big(1-
   \mu \int_{0}^t e^{-\gamma s -\mu\int_{0}^s q(u_l(v)) \dd l }  q( u_s(v)) \dd s\big)\Big) \Delta v
\end{align*}
Hence
$
            D_v Z^v(t) =
    e^{\gamma t +\mu\int_{0}^t  q( u_s(v)) \dd s }\big(1-
   \mu \int_{0}^t e^{-\gamma s -\mu\int_{0}^s q(u_l(v)) \dd l }  q( u_s(v)) \dd s\big).
$
\end{proof}
\begin{lemma}[Bound on the derivative]
\label{lemma:bound}
For all $t\geq 0$, $\frac{\partial Z^v(t)}{ \partial v } \geq e^{\gamma t}$.
\end{lemma}
\begin{proof}
By Lemma~\ref{eq:lemma_Taylor},
\begin{align*}
            &\frac{\partial Z^v(t)}{ \partial v } =
    e^{\gamma t +\mu\int_{0}^t  q( u_s(v)) \dd s }\big(1-
   \mu \int_{0}^t e^{-\gamma s -\mu\int_{0}^s q(u_l(v)) \dd l }  q( u_s(v)) \dd s\big)\\
   &\geq
    e^{\gamma t +\mu\int_{0}^t  q( u_s(v)) \dd s }\big(1-
\int_{0}^t e^{ -\mu\int_{0}^s    q(u_l(v)) \dd l } \mu  q( u_s(v)) \dd s\big)\\
&=
    e^{\gamma t +\mu\int_{0}^t  q( u_s(v)) \dd s }\big(1+
[ e^{ -\mu\int_{0}^s    q(u_l(v)) \dd l }]_{0}^t\big)\\
&=    e^{\gamma t +\mu\int_{0}^t  q( u_s(v)) \dd s }\big(
e^{ -\mu\int_{0}^t    q(u_l(v)) \dd l }\big)
=e^{\gamma t }.
\end{align*}
\end{proof}
\begin{corollary}[Uniqueness of the stable solution]
\label{corrolary:uniqueness_stable}
There can be at most one $y_0\in\RR$ such that $ \lim_{t\to +\infty} Z^{y_0}(t) $ is finite.
\end{corollary}
\begin{proof}
We deduce from Lemma~\ref{lemma:bound}, that for any $t$, and any $v_1 < v_2$,
$ Z^{v_2}(t)  >  Z^{v_1}(t) + (v_2 - v_1) e^{\gamma t}   $.
Now, if $Z^{v_1}(t)$ converges to a finite value, then  $Z^{v_2}(t)$ goes to $+\infty$. This guarantees that $\lim_{t\to +\infty} Z^v(t)$ is finite for at most one $v$.
\end{proof}
\subsection*{Lemma~\ref{lemma:algo}}
\begin{proof}
    The value function satisfies $\mathcal{F}_{V_0}$, and is bounded and positive.
 By Corollary~\ref{corrolary:uniqueness_stable}, it is the only  bounded solution. 
\end{proof}
\subsection*{Theorem~\ref{theorem:algo}}
\begin{proof}
    By corollary~\ref{corrolary:uniqueness_stable}, the condition ``$t\to y^{(n)}(t)$ is diverging to $+\infty$" is equivalent to  $y^{(n)}(0)> V_0^\star$. 
Hence, if $V_0^\star$ is in $[a,b]$ at step $(n)$, it will also be in $[a,b]$ at step $(n+1)$. 
\end{proof}
\section*{Proof of Section 4}
Let  $T_w$ the random variable``age at the time of the next won auction".
We  set
$
   A(t) = \int_0^t \gamma + \mu q(b(s)) \dd s.
$
We also define the random variable $N_b$ as the  number of won auctions until the end of the sequence $T_\gamma$. We may now state how to compute $V_b$ and its initial value:
\begin{lemma}
\label{lemma:solution_ode_bid}
    The solution $V_b$ is given by:
    \begin{align}
         V_b(\tau) &= e^{A(\tau)} \left( V_b(0) - \int_0^\tau e^{-A(t)}\cdot \mu \cdot U( V_b(0) + k(t) , b(t))    dt \right) \label{eq:vb_integral} \\
         V_b(0) &= \EE(N_b) \times \EE \left( U( k(T_w), b(T_w)  )  | T_w < T_\gamma \right).   \label{eq:v0}
    \end{align}
\end{lemma}
\begin{proof}
   Equation \eqref{eq:vb_integral} is the integral solution of the linear differential equation of Lemma~\ref{lemma:ODE}. 
 Equation  \eqref{eq:v0} states that $V(0)$ is precisely the expected number of won auctions multiplied by the average return on those won auctions, which seem intuitively obvious. We may prove it more formally as follows, by first rewriting \eqref{eq:vb_integral} as :
    $$  e^{-A(\tau)} \cdot V_b(\tau)  = V_b(0) \cdot \left( 1 - \int_0^\tau e^{-A(t)}\cdot \mu \cdot q(b(t)) dt   \right) - \int_0^\tau e^{-A(t)}\cdot \mu \cdot U\left( k(t) , b(t)  \right)    dt \big) $$
Since we know that $ V_b(\tau)$ is bounded, and $ A(\tau) > \gamma.\tau \xrightarrow{ \tau \to \infty } +\infty  $, the left-hand side goes to $0$ when $\tau \to+\infty$ 
Thus
\begin{align*}
    V_b(0)  &= \frac{ \int_0^\infty e^{-A(t)}\cdot \mu \cdot U\left( k(t) , b(t)  \right)    dt  }{  1 - \int_0^\infty e^{-A(t)}\cdot \mu \cdot q(b(t)) dt   } \\
     &=  \frac{ \int_0^\infty e^{-A(t)}\cdot \mu \cdot U\left( k(t) , b(t)  \right)    dt  }{  \Pr(T_w < T_\gamma)   } \times \frac{\Pr(T_w < T_\gamma) }{1 -\Pr(T_w < T_\gamma)  }. \label{eq:v0_integral}
\end{align*}
Here we used the formula
$\int_0^\infty e^{-A(t)}\cdot \mu \cdot q(b(t)) dt = \Pr(T_w < T_\gamma)  $.
We conclude by identifying the two factors on the right side of equation \eqref{eq:v0_integral} as some expectations:
$$ \EE \left( U( k(T_w), b(T_w)  )  | T_w < T_\gamma \right)  =  \frac{ \int_0^\infty e^{-A(t)}\cdot \mu \cdot U\left( k(t) , b(t)  \right)    dt  }{  \Pr(T_w < T_\gamma)   }  $$
is the average return per won auction, and
$ \EE(N_b) = \frac{\Pr(T_w < T_\gamma) }{1 -\Pr(T_w < T_\gamma)  } $
is the expected number of won auctions. (We note that it follows a geometric law of parameter $\Pr(T_w < T_\gamma)$, hence the formula.)
\end{proof}
\paragraph{Limit when $\gamma$ goes to 0} 
Since the process ends after a time $T_\gamma$ following an exponential law, the expected lifetime is $\nicefrac{1}{\gamma}$.
Intuitively, increasing the length of the sequence should increase linearly the expected return, as with the same bidding strategy we will increase the number of won items without changing their average price or value. The following lemma confirms this intuition for long sequences, by proving that the average return per time unit, $\gamma\cdot V_{b,\gamma}(0)$ is converging:
\begin{lemma}
\label{lemma:limit_gamma}
\begin{align}
   \lim_{\gamma\to 0} \gamma V_{b,\gamma}(0) = \mu \frac{
      \int_0^{+\infty}e^{-\mu\int_0^t q(b_s)\dd s }U(k_t,b_t)\dd t 
   }
   {
   \int_0^{+\infty}e^{-\mu\int_0^t q(b_s)\dd s }\dd t 
   } 
\end{align}
\end{lemma}
\begin{proof}
We have 
$
   \Pr(T_w < T_\gamma) = \int_0^{+\infty}e^{-\gamma t } e^{-\int_0^t q(b_s)\dd s }\mu q(b_t))\dd t= 
   1- \gamma\int_0^{+\infty}e^{-\mu\int_0^t q(b_s)\dd s-\gamma t  }\dd t 
$ by integration by parts.
We then take the limit in Equation~\eqref{eq:v0_integral}.
\end{proof}
\subsection*{Theorem~\ref{th:special_case}}
\begin{proof}
    From Lemma~\ref{lemma:limit_gamma}, 
we have:
\begin{align*}
  \lim_{\gamma\to 0 }  \gamma V_0^\alpha  =( 1 - 1/2\alpha )\frac{ \int_0^\infty \alpha\mu k_t^2 e^{-(\mu\alpha\int_0^t k_s\dd s) }\dd t}{\int_0^\infty e^{-(\mu\alpha\int_0^t k_s\dd s) }\dd t}.
\end{align*}
Since  $k_t=1-\frac{1}{1+t}$, we have $ \int_0^t k_s\dd s= t -\ln(1+t)$
and
 $e^{-(\mu\alpha\int_0^t k_s\dd s) }=   e^{-\alpha \mu t}(1+t)^{\alpha \mu}$.
For  $m\geq 0$, we have
 \begin{align*}
& \int_0^{+\infty} (1+t)^{-m} e^{-\alpha \mu t}(1+t)^{\alpha \mu}\\
 &=\int_0^{+\infty}  e^{-\alpha \mu t}(1+t)^{\alpha \mu-m}\dd t 
 \\&=\int_1^{+\infty}  e^{-\alpha \mu (x-1)}x^{\alpha \mu-m}\dd x \\
&=e^{\alpha \mu }\int_1^{+\infty}  e^{-\alpha \mu x}x^{\alpha \mu-m}\dd x\\
\intertext{we set $\theta = (\alpha\mu)^{-1}$, $k =(\alpha \mu-m+1) $,$\tilde{\alpha}  =\alpha \mu$}
&=e^{\tilde{\alpha}}\int_1^{+\infty} x^{k-1} e^{-\frac{x}{\theta }}\dd x\\
&=e^{\tilde{\alpha}}\theta^{k}\int_{1/\theta}^{+\infty} t^{k-1} e^{-t}\dd t
\\&=e^{\tilde{\alpha}}\theta^{k}\Gamma(k,1/\theta)
=e^{\tilde{\alpha}}\tilde{\alpha}^{m-1-\tilde{\alpha}}\Gamma(\tilde{\alpha}-m+1,1/\theta).
\end{align*}
Therefore
\begin{align*}
&( 1 - 1/2\alpha )\frac{ \int_0^\infty \alpha\mu k_t^2 e^{-(\mu\alpha\int_0^t k_s\dd s) }\dd t}{\int_0^\infty e^{-(\mu\alpha\int_0^t k_s\dd s) }\dd t}\\
&=\alpha\mu( 1 - 1/2\alpha )\frac{\tilde{\alpha}^{-1-\tilde{\alpha}}\Gamma(\tilde{\alpha}+1,1/\theta)+\tilde{\alpha}^{1-\tilde{\alpha}}\Gamma(\tilde{\alpha}-1,1/\theta)-2\tilde{\alpha}^{-\tilde{\alpha}}\Gamma(\tilde{\alpha},1/\theta)}{\tilde{\alpha}^{-1-\tilde{\alpha}}\Gamma(\tilde{\alpha}+1,1/\theta)}\\
&=( 1 - 1/2\alpha )\frac{\tilde{\alpha}\Gamma(\tilde{\alpha}+1,1/\theta)+\tilde{\alpha}^{3}\Gamma(\tilde{\alpha}-1,1/\theta)-2\tilde{\alpha}^{2}\Gamma(\tilde{\alpha},1/\theta)}{\Gamma(\tilde{\alpha}+1,1/\theta)}\\
&=( 1 - 1/2\alpha )\frac{\tilde{\alpha}(\Gamma(\tilde{\alpha}+1,1/\theta)-\tilde{\alpha}\Gamma(\tilde{\alpha},1/\theta))-\tilde{\alpha}^{2}(\Gamma(\tilde{\alpha},1/\theta)
-\tilde{\alpha}\Gamma(\tilde{\alpha}-1,1/\theta))
}{\Gamma(\tilde{\alpha}+1,1/\theta)}
\end{align*}
Now, we use the property
$\Gamma(s+1,1/\theta) = s\Gamma(s,1/\theta) + \theta^{-s} e^{-1/\theta}$
to get 
\begin{align*}
  &  ( 1 - 1/2\alpha )\frac{\tilde{\alpha}(\Gamma(\tilde{\alpha}+1,1/\theta)-\tilde{\alpha}\Gamma(\tilde{\alpha},1/\theta))-\tilde{\alpha}^{2}(\Gamma(\tilde{\alpha},1/\theta)
-(\tilde{\alpha}-1)\Gamma(\tilde{\alpha}-1,1/\theta))
+\tilde{\alpha}^{2}\Gamma(\tilde{\alpha}-1,1/\theta))
}{\Gamma(\tilde{\alpha}+1,1/\theta)}\\
&= ( 1 - 1/2\alpha )\frac{\tilde{\alpha}(\theta^{-\tilde{\alpha}}e^{-1/\theta})-\tilde{\alpha}^{2}(\theta^{-(\tilde{\alpha}-1)}e^{-1/\theta})+
\tilde{\alpha}^{2}\Gamma(\tilde{\alpha}-1,1/\theta))
}{\Gamma(\tilde{\alpha}+1,\tilde{\alpha})}\\
&= ( 1 - 1/2\alpha )\frac{\tilde{\alpha}(\tilde{\alpha}^{\tilde{\alpha}}e^{-\tilde{\alpha}})-\tilde{\alpha}^{2}(\tilde{\alpha}^{\tilde{\alpha}-1}e^{-\tilde{\alpha}})
+\tilde{\alpha}^{2}\Gamma(\tilde{\alpha}-1,\tilde{\alpha})}{\Gamma(\tilde{\alpha}+1,\tilde{\alpha})}\\
&= ( 1 - 1/2\alpha )\tilde{\alpha}^{2}\frac{
\Gamma(\tilde{\alpha}-1,\tilde{\alpha})}{\Gamma(\tilde{\alpha}+1,\tilde{\alpha})}.
\end{align*}
\end{proof}
\subsection*{Theorem~\ref{theorem:limit}}
\begin{proof}
    Let a smooth bid policy $b$ satisfying $b_t\leq k_t$. We have
$U(k_t,b_t)=q(b_t)(k_t-\mathfrak{p}(b_t))$. Set $Q_t = \int_0^t q(b_s)\dd s $. Then by integration by part, and because $0=k_0=b_0$ by assumption,
\begin{align*}
     &\int_0^{+\infty}\mu e^{-\mu\int_0^t q(b_s)\dd s }U(k_t,b_t)\dd t \\
     &=\int_0^{+\infty}\mu e^{-\mu Q_t }q(b_t)(k_t-\mathfrak{p}(b_t))\dd t\\& =
     \int_0^{+\infty} e^{-\mu Q_t }(\dot{k}_t-\dot{b}_t\dot{\mathfrak{p}}(b_t))\dd t.
\end{align*}
Now observing that 
\begin{align*}
  \lim_{\mu\to+\infty}       \frac{\int_0^{+\infty} e^{-\mu Q_t }(\dot{k}_t-\dot{b}_t\dot{\mathfrak{p}}(b_t))\dd t}
         {\int_0^{+\infty} e^{-\mu Q_t }\dd t} = \dot{k}_0-\dot{b}_0\dot{\mathfrak{p}}(b_0) = \dot{k}_0-\dot{b}_0\dot{\mathfrak{p}}(0),
\end{align*}
and combining with  Lemma~\ref{lemma:limit_gamma}, we get 
\begin{align*}
 \lim_{\mu\to+\infty}  \lim_{\gamma\to 0} \gamma V_{b,\gamma,\mu}(0) = \dot{k}_0-\dot{b}_0\dot{\mathfrak{p}}(0)
\end{align*}
We see that only the derivative of $b$ at $0$ matters here, so we can without loss of generality suppose that $b_t=\alpha k_t$ for some $\alpha\leq 1$.
We end up with 
 $\dot{k}_0(1-\alpha\dot{\mathfrak{p}}(0))$. Since $\dot{\mathfrak{p}}(0)>0$, this quantity is minimized for $\alpha =1$ and maximized for $\alpha=0$, otherwise said, 
the left-hand side of the Theorem is equal to  $\dot{\mathfrak{p}}(0)$.
Then we observe that  $\mathfrak{p} = \frac{\int_0^t s \dd (q(s))}{q(s) }$. The remaining follows after integration by part and taking the limit at 0.
\end{proof}

\bibliographystyle{plain}

\end{document}